\documentclass[draftcls, onecolumn, 11pt]{IEEEtran}
\usepackage{amsmath,amssymb,threeparttable,placeins, enumerate ,caption}
\usepackage{mathtools,relsize}
\usepackage[pdftex]{graphicx}
\usepackage{epstopdf}
\usepackage{color}
\usepackage[usenames]{xcolor}
\usepackage{amsfonts,adjustbox}
\usepackage{latexsym}
\usepackage{subfigure, multirow,makecell,stfloats}
\usepackage{lscape}
\usepackage[figuresright]{rotating}
 \usepackage{array}

\usepackage{algorithm}
\usepackage{cite}
\usepackage[noend]{algpseudocode}

\def\qi#1 {\fbox {\footnote {\ }}\ \footnotetext { From Qi: {\color{red}#1}}}
\usepackage{amssymb}

\usepackage{amsthm}

\theoremstyle{remark} 
\newtheorem{theorem}{{{\textit{Theorem}}}}

\newtheorem{lemma}{{{\textit{Lemma}}}}
\newtheorem{corollary}{{{{\textit{Corollary}}}}}

\newtheorem{remark}{{{\textit{Remark}}}}

\newtheorem{example}{{{\textit{Example}}}}

\newtheorem{construction}{{{\textit{Construction}}}}
\title{New Constructions of Locally Perfect Nonlinear Functions and Their Application to Sequence Sets With Low Ambiguity Zone}

	\author{Zhiye Yang, Zheng Wang, Huaning Liu, and Keqin Feng 
	\thanks{Zhiye Yang and Huaning Liu are with the Research Center for Number Theory and Its Applications
School of Mathematics, Northwest University, Xi’an, 710127, Shaanxi, China.
(e-mail: zyyang02@126.com, hnliu@nwu.edu.cn)

		Zheng Wang is with the School of Mathematics, Southwest Jiaotong University, Chengdu, 611756, China. (e-mail: wang\_z@my.swjtu.edu.cn).

		Keqin Feng is with the Department of Mathematical Sciences, Tsinghua University, Beijing, 100084, China.
(e-mail: fengkq@tsinghua.edu.cn)

	}
}

\begin{document}
	\maketitle
	\begin{center}
		\textbf{Abstract}
	\end{center}

	Low Ambiguity Zone (LAZ) sequences play a pivotal role in modern integrated sensing and communication (ISAC) systems. Recently, Wang \textit{et al.} [arXiv:2501.11313] proposed a definition of locally perfect nonlinear functions (LPNFs) and constructed three classes of both periodic and aperiodic LAZ sequence sets with flexible parameters by applying such functions and interleaving techniques. Some of these LAZ sequence sets are asymptotically optimal with respect to the Ye-Zhou-Fan-Liu-Lei-Tang bounds under certain conditions. In this paper, we present constructions of three new classes of LPNFs with new parameters.  Based on these LPNFs, we further propose a series of LAZ sequence sets that offer more flexible parameters. Furthermore, our results show that some of these classes are asymptotically optimal in both the periodic and aperiodic cases, respectively.

	\begin{center}
		\textbf{Index Terms}
	\end{center}
	
Ambiguity function (AF), Low Ambiguity Zone (LAZ), Locally Perfect Nonlinear Function (LPNF), Sequence set.

	\section{Introduction}
	
	To design sequence sets with various capabilities is one of the significant topics in the research field of communication and radar systems. A critical challenge now lies in the design of sequences that exhibit desirable correlation properties not only in the time domain but also
		across phase shifts, as effectively captured by the ambiguity function (AF). The AF is a important metric for evaluating the communication and sensing capabilities of sequences\cite{GolombG}, sequences exhibiting zero or low AF sidelobes are essential for mitigating the Doppler effect in advanced integrated sensing and communication (ISAC) systems, enabling accurate signal detection and processing in highly dynamic environments\cite{Hassanien}.
    
	In 2013, Ding and Feng \textit{et al.}\cite{DingFeng2013} investigated the theoretical lower bound of the maximum sidelobe of the AF based on the Welch bound. Additionally, they presented sequences possessing favorable AF characteristics. Notably, a single sequence among these achieves the desired theoretical lower bound across the entire delay-Doppler (DD) region. Moreover, Wang and Gong \textit{et al.}\cite{Gong2013, WangGong2011, WangGong2013} proposed several classes of sequences and sequence sets with low AF magnitudes in the entire DD region, leveraging the excellent algebraic properties of additive and multiplicative characters over finite fields. In addition to constructing sequences with low AF sidelobes based on finite field theory, researchers have also attempted to design sequences through optimization algorithms \cite{Soltanalian12, Cui13, Stoica14, Razaviyayn15}. However, some of these sequences exhibit low AF sidelobes only in specific DD regions.
	
	In fact, for many practically interesting applications in the ISAC system, the Doppler frequency range is typically much narrower than the bandwidth of the transmitted signal \cite{duggal2020doppler, kumari}. Therefore, it's likely that there's no need to take into account the entire duration of the signal. In other words, in the study of different applications, it is usually sufficient to focus on the AF sidelobes in the region relevant to specific application requirements. 
	Based on this idea, in the past five years, people have made breakthroughs in the research of AF. 
    In \cite{Ye2022}, Ye \textit{et al.} first introduced a new definition regarding the low AF sidelobes in a specific region, termed the low/zero ambiguity zone (LAZ/ZAZ). Therefore, they referred to the sequence sets that meet the requirements of the LAZ/ZAZ as LAZ/ZAZ sequence sets, and derived the theoretical lower bounds for the periodic and aperiodic AF sidelobes of unimodular sequence sets in the LAZ, respectively. These bounds are collectively known as the Ye--Zhou--Fan--Liu--Lei--Tang bounds. For brevity, we refer to them as the ``YZFLLT bounds'' throughout this paper.
    By applying cubic sequences, certain quadratic phase sequences, and cyclic difference sets, they also presented four types of the most suitable constructions of sequence families, and these sequences achieve the lower bounds. Meng \textit{et al.} \cite{meng2024new} recently established a tighter aperiodic bound compared to [28, Theorem 4] under specific LAZ sequence parameters.
    Subsequently, Cao \textit{et al.}\cite{Cao2022} identified that a binary sequence set proposed in \cite{ZhouTangGong2008} possesses low ambiguity characteristics over a large area. Note that these sequences are interleaved sequences. Thereafter, Tian \textit{et al.}\cite{Tian2024} defined a set of LAZ sequences based on the perfect nonlinear functions (PNFs). Very recently, Wang \textit{et al.}\cite{WangZhou2025} generalized the concept of PNFs, introducing locally perfect nonlinear functions (LPNFs), and presented three constructions of LPNFs. And then, by directly applying these LPNFs and the interleaving techniques, they gave a series of LAZ sequence sets with flexible parameters. Such sequence sets with selecting appropriate parameters are asymptotically optimal with respect to YZFLLT bounds in both periodic and aperiodic cases, respectively. In addition, Wang \textit{et al.} \cite{Wang25_2} proposed an asymptotically optimal LAZ sequence set based on a cubic function.
    Consequently, this drives us to develop sequences with zero or low AF sidelobe characteristics in desired regions. 
	The specific parameters of the above mentioned sequences are listed in Table \ref{Table1}. In column 6 of Table \ref{Table1}, P.O, P.AO, and AP.AO. mean optimal with periodic, asymptotically optimal with periodic, and asymptotically optimal with aperiodic, respectively.
	\begin{table}[h]
		\centering
		\caption{Asymptotically Optimal or Optimal LAZ/ZAZ sequence sets}
		\label{Table1}
		\begin{adjustbox}{max width=\textwidth}
		\begin{tabular}{cccccccc}
			\Xhline{1pt}
			Ref.  & Length& Set size & $Z_x$ & $Z_y$ & ${\theta}_{\text{max}}(\hat{\theta}_{\text{max}})$ & Optimality &  \,\,Constraints\\
			\Xhline{1.5pt}
			\cite{DingFeng2013} & $q-1$      & $1$ & $q-1$ & $q-1$ & $\sqrt{q}$ & P.O.  & $q=p^l$, $\textrm{$p$ is a prime}$.\\
			\cite{Ye2022} & $p$      & $1$ & $p$ & $p$ & $\sqrt{p}$ & P.O.  &  $\textrm{$p$\,is an odd prime}$.\\
			\cite{Ye2022} & $L$      & $1$ & $\frac{L}{r}$ & $r$ & $0$ & P.O.  &   \makecell[c]{$\gcd(a,L)=1\,\textrm{if $L$ is odd },$\\$r=\gcd(2a,L)>1$.}\\
			\cite{Ye2022} & $L$      & $N$ & $\frac{\lfloor L/N\rfloor}{r}$ & $r$ & $0$ & \makecell[c]{P.O, \\if $N|L$.}  &   \makecell[c]{$\gcd(a,L)=1\,\textrm{if L is odd },$\\$r=\gcd(2a,L)>1$.}\\
			\cite{Tian2024} & $MN^2$      & $MN$ & $\lfloor \frac{N}{K}\rfloor$ & $K$ & $0$ & P.AO.  & $K<N$, $\gcd(K,N)=1$.\\
			\cite{Tian2024} & $N(KN+P)$      & $N$ & $N$ & $K$ & $0$ & P.AO.  & $\gcd(P,NK)=1$.\\
			\cite{Tian2024} & $p(p-1)$      & $p$ & $p-1$ & $p$ & $\sqrt{p}$ & P.AO.  & $\textrm{$p$\,is an odd prime}$.\\
            \cite{Wang25_2} & $N$      & $M$ & $p$ & $\lfloor\frac{N}{M}\rfloor$ & $\sqrt{N}$& \makecell[c]{P.AO}  & \makecell[c]{$N$ \textrm{\,is an odd number,} $p$ \textrm{is the} \\      \textrm{smallest prime of} $N$, $M\leq N$.}\\
			\cite{WangZhou2025} & $N^2$      & $N$ & $p$ & $N$ & \makecell[c]{$N$,\\$N+p-1$} & \makecell[c]{P.AO,\\AP.AO.}  & \makecell[c]{$N$ \textrm{\,is an odd number,}\\ $p|N(\textrm{the smallest prime})$.}\\
			\cite{WangZhou2025} & $NK$      & $N$ & $p$ & $K-N+1$  & \makecell[c]{$K$,\\$K+p-1$} & \makecell[c]{P.AO,\\AP.AO.}  & \makecell[c]{$N$ \textrm{\,is an odd number,} \\$N<K<2N-1$,\\$p|N(\textrm{the smallest prime})$.}\\
			\makecell[c]{$\textrm{Theorem}$\,\ref{optimalc1} (i),\\$\textrm{Theorem}$\,\ref{optimalc2} (i)}& $NK$      & $N$ & $p-1$ & $K$ & \makecell[c]{$K$,\\$K+p-1$} & \makecell[c]{P.AO,\\AP.AO.}  & \makecell[c]{$K=p^e, \textrm{$p$\,is an odd prime},$\\$N=\varphi(p^e)$.}\\
            			$\textrm{Theorem}$\,\ref{optimalc1} (ii)& $NK$      & $N$ & $N$ & $p$& $K+N-1$ & \makecell[c]{P.AO}  & \makecell[c]{$K=p^e, \textrm{$p$\,is an odd prime},$\\$N=\varphi(p^e)$.}\\
                       $\textrm{Theorem}$\,\ref{optimal3c} & $NK_1$      & $N$ & $p-1$ & $K_1-K+2$ & \makecell[c]{$K_1$,\\$K_1+p-2$}& \makecell[c]{P.AO,\\AP.AO.}  & \makecell[c]{$K=p^e, \textrm{$p$\,is an odd prime},$\\$K\leq K_1<2 K$, $N=\varphi(p^e)$.}\\
			\Xhline{1pt}
		\end{tabular}
		\end{adjustbox}
	\end{table}


	In this paper, we show some generic new constructions on LPNFs in Section 3, and by using these new LPNFs, we present a series of LAZ
	sequence sets with new parameters. Some sequence sets of the proposed are asymptotically optimal with respect to YZFLLT bounds both periodic and aperiodic cases, respectively. Before doing our main results, in Section 2, we introduce some preliminaries on LAZ sequence sets and LPNFs, including
	some related results we need in this paper.

	\section{LAZ Sequence Sets and Locally Perfect Nonlinear Functions}
	In this section, we review some fundamental concepts related to LAZ sequence sets and LPNFs, along with several results that will be used in the subsequent analysis.

	\subsection{LAZ Sequence Sets}\label{GR}
	
	Throughout the rest of the paper, let $\mathbb{C}$ denote the set of all complex numbers and $\mathcal{S}$ denote the set of sequences of length $D$ which has a size of $M$, all sequences
	$\mathbf{a}=(a(0),a(1),a(2),\cdots,a(D-1))$ are regarded as \textit{unimodular sequences} with length $D$.
	Specifically, ${a}(k)\in \mathbb{C}$ and $|{a}(k)|=1$ for all $0\leq k\leq {D-1}$. In particular, $D$ is the period of $\mathbf{a}$ if ${a}(k+D)={a}(k)$ for all $k\geq0$. $\omega_{N}=e^{\frac{2\pi\sqrt{-1}}{N}}=e^{\frac{2\pi i}{N}}$ is a $N$-th primitive complex root of unity.
	
	Let $\mathbf{a}$ and $\mathbf{b}$ be unimodular sequences with length $D$. Then the \textit{aperiodic cross-ambiguity function (cross-AF)} of $\mathbf{a}$ and $\mathbf{b}$ at time shift $\tau$ and Doppler shift $\nu$ is defined by:
	\begin{equation}\nonumber\label{APAF}
		\widehat{AF}_{\mathbf{a},\mathbf{b}}(\tau,\nu)=\left\{\begin{array}{cl}
			\sum\limits_{t=0}^{D-1-\tau}a(t)\overline{b}(t+\tau)\omega_{D}^{\nu t}, &0\leq\tau\leq D-1, \\
			\sum\limits_{t=-\tau}^{D-1}a(t)\overline{b}(t+\tau)\omega_{D}^{\nu t}, &-(D-1)\leq\tau<0 , \\
			0, &|\tau|\geq D.
		\end{array}\right.
	\end{equation}
	where for $c\in\mathbb{C}$, $\overline{c}$ represents complex conjugate of $c$.\\
	If both of $\mathbf{a}$ and $\mathbf{b}$ are periodic with period $D$, the \textit{periodic cross-AF} of $\mathbf{a}$ and $\mathbf{b}$ is defined by:
	\begin{equation*}\label{PAF}
		{AF}_{\mathbf{a},\mathbf{b}}(\tau,\nu)=\sum\limits_{t=0}^{D-1}a(t)\overline{b}(t+\tau)\omega_{D}^{\nu t},
	\end{equation*}
	where $t+\tau$ to be taken modulo $D$. We denote $\widehat{AF}_{\mathbf{a}}(\tau,\nu)=\widehat{AF}_{\mathbf{a},\mathbf{a}}(\tau,\nu)$ and ${AF}_{\mathbf{a}}(\tau,\nu)={AF}_{\mathbf{a},\mathbf{a}}(\tau,\nu)$, called the \textit{auto-AF} of $\mathbf{a}$, aperiodic and periodic, respectively.
	
	For a sequence set $\mathcal{S}$ that consists of $M$ sequences of length $D$, the \textit{maximum periodic AF magnitude} of $\mathcal{S}$  over a region $\Pi=(-Z_x,Z_x)\times(-Z_y,Z_y)\subseteq(-D,D)\times(-D,D)$ is defined as $$\theta_{\text{max}}(\mathcal{S})=\{\theta_{\text{A}}(\mathcal{S}),\theta_{\text{C}}(\mathcal{S})\},$$ where $\theta_{\text{A}}(\mathcal{S})$ denotes the maximum periodic auto-AF magnitude by
	\begin{equation*}\label{PAF}
		\theta_{\text{A}}(\mathcal{S})={\text{max}}\{|{AF}_{\mathbf{a}}(\tau,\nu)|:\mathbf{a}\in\mathcal{S},(0,0)\neq(\tau,\nu)\in\Pi\}
	\end{equation*}
	and $\theta_{\text{C}}(\mathcal{S})$ denotes the maximum periodic cross-AF magnitude by
	\begin{equation*}\label{PAF}
		\theta_{\text{C}}(\mathcal{S})={\text{max}}\{|{AF}_{\mathbf{a},\mathbf{b}}(\tau,\nu)|:\mathbf{a}\neq\mathbf{b}\in\mathcal{S},(\tau,\nu)\in\Pi\}
	\end{equation*}
	Such a sequence set $\mathcal{S}$ is denoted by $(M, D,\Pi,\theta_{\text{max}})$-LAZ periodic sequence set, where $M$ is the set size, $D$ is the length of sequence, $\Pi$ is the low ambiguity zone, as well as $\theta_{\text{max}}$ denotes the maximum periodic AF magnitude in region $\Pi$.
	
	For aperiodic case, we define similarly $\hat{\theta}_{\text{max}}(\mathcal{S})=\{\hat{\theta}_{\text{A}}(\mathcal{S}),\hat{\theta}_{\text{C}}(\mathcal{S})\},$ to be the maximum aperiodic AF magnitude over a region $\Pi$, where
	\begin{equation*}\label{PAF}
		\hat{\theta}_{\text{A}}(\mathcal{S})={\text{max}}\{|\widehat{AF}_{\mathbf{a}}(\tau,\nu)|:\mathbf{a}\in\mathcal{S},(0,0)\neq(\tau,\nu)\in\Pi\}
	\end{equation*}
	\begin{equation*}\label{PAF}
		\hat{\theta}_{\text{C}}(\mathcal{S})={\text{max}}\{|\widehat{AF}_{\mathbf{a},\mathbf{b}}(\tau,\nu)|:\mathbf{a}\neq\mathbf{b}\in\mathcal{S},(\tau,\nu)\in\Pi\}\,\,\,\,\,
	\end{equation*}
	Such a sequence set $\mathcal{S}$ is denoted by $(M,D,\Pi,\hat{\theta}_{\text{max}})$-LAZ aperiodic sequence set.
	
	As mentioned, in \cite{Ye2022}, authors established lower bounds $\Delta$ and $\hat{\Delta}$ of $\theta_{\text{max}}$ and $\hat{\theta}_{\text{max}}$, respectively. The following theorem is a consequence of \cite{Ye2022}.

	\begin{lemma} [Ye--Zhou--Fan--Liu--Lei--Tang bounds \cite{Ye2022}] \label{2013bound}
		For any unimodular periodic or aperiodic sequence set: $\mathcal{S}=(M,D,\Pi,\theta_{\text{max}})$ or $\mathcal{S}=(M,D,\Pi,\hat{\theta}_{\text{max}})$, where $\Pi=(-Z_x,Z_x)\times(-Z_y,Z_y)$. We have
		
		(i)\,\,$\theta_{\text{max}}\geq\Delta$, where $\Delta=\frac{D}{\sqrt{Z_y}}\sqrt{\frac{MZ_xZ_y-D}{D(MZ_x-1)}}$;
		
		(ii)\,\,$\hat{\theta}_{\text{max}}\geq\hat{\Delta}$, where $\hat{\Delta}=\frac{D}{\sqrt{Z_y}}\sqrt{\frac{MZ_xZ_y-D-Z_x+1}{(MZ_x-1)(D+Z_x-1)}}$.
	\end{lemma}
	Further assume that $\rho_{\text{LAZ}}=\frac{\theta_{\text{max}}}{\Delta}(\geq1)$ and $\hat{\rho}_{\text{LAZ}}=\frac{\hat{\theta}_{\text{max}}}{\hat{\Delta}}(\geq1)$, $\Delta$ and $\hat{\Delta}$ are YZFLLT bounds, $\rho_{\text{LAZ}}$ and $\hat{\rho}_{\text{LAZ}}$ are conventionally referred to as \textit{optimality factors}. The optimality factors $\rho_{\text{LAZ}}$ and $\hat{\rho}_{\text{LAZ}}$ serve as metrics to evaluate how closely the constructed sequences approach the theoretical lower bounds in both periodic and aperiodic cases, respectively. i.e.
	
	$\bullet$\,\,If ${\rho_{\text{LAZ}}}=1$ or $\hat{\rho}_{\text{LAZ}}$=1, the sequence set $\mathcal{S}$ is called \textit{optimal} in periodic and aperiodic case, respectively;
	
	$\bullet$\,\,If $\lim\limits_{D \to \infty} \rho_{\text{LAZ}}=1$ or $\lim\limits_{D \to \infty} \hat{\rho}_{\text{LAZ}}=1$, the sequence set $\mathcal{S}$ is called \textit{asymptotically optimal} in periodic and aperiodic case, respectively.

\begin{remark}
 When the parameters $M$, $Z_x$, and $Z_y$ satisfy either $Z_x > \sqrt{\frac{3 N^2}{M Z_y}}$ with $M Z_y \geq 3$, or $Z_x > \frac{\pi}{\gamma}$ with $5 \leq M Z_y \leq N^2$, where $\gamma = \arccos\left(1 - \frac{M Z_y}{N^2}\right)$, a tighter aperiodic AF lower bound was derived by Meng \textit{et al.} in \cite[Corollaries 2 and 3]{meng2024new}, improving upon the result given in \cite[Theorem 4]{Ye2022}. Nevertheless, since our constructed LAZ sequence set does not fulfill either of these conditions, we adopt the lower bound for aperiodic AF provided by Ye \textit{et al.} in \cite[Theorem 4]{Ye2022}.
\end{remark}
	
	
	\subsection{Locally Perfect Nonlinear Functions}\label{LPNF}
	
	Let $K\geq N$ be two positive integers, and $f:$ $\mathbb{Z}_N \rightarrow \mathbb{Z}_K$ be a function, where $\mathbb{Z}_N=\mathbb{Z}/N\mathbb{Z}$. For $0<Z_x\leq N$ and $0<Z_y\leq K$, The function $f$ is called a \textit{locally perfect nonlinear function} and denoted by $\langle N,K,Z_x,Z_y\rangle$-LPNF if for any $-Z_y<b<Z_y$ and $-Z_x<a<Z_x$, $a\neq0$,
	the equation $f(x+a)-f(x)=b$ has at most one solution $x\in\mathbb{Z}_N$.
	
	\begin{remark}
If \( Z_x \geq \left\lceil \frac{N+1}{2} \right\rceil \), then the interval \( -Z_x < a < Z_x \) covers all elements \( a \in \mathbb{Z}_N \). Similarly, if \( Z_y \geq \left\lceil \frac{K+1}{2} \right\rceil \), then \( -Z_y < b < Z_y \) covers all elements \( b \in \mathbb{Z}_K \).
	\end{remark}
	
	The notion of LPNF was introduced in \cite{WangZhou2025} and successfully used to construct LAZ sequence sets by the interleaving technique. The result is as follows.
	
	\begin{lemma}[\cite{WangZhou2025} Theorems 1 and 2]\label{2025sset}
	From an \( \langle N, K, Z_x, Z_y \rangle \)-LPNF, a LAZ sequence set can be constructed with parameters \( (M, D, \Pi, \theta_{\max})\) or \( (M, D, \Pi, \hat{\theta}_{\max}) \), where \( M = N \), \( D = NK \), \( \Pi = (-Z_x, Z_x) \times (-Z_y, Z_y) \), \( \theta_{\max} = K \), and \( \hat{\theta}_{\max} = K + Z_x - 1 \).
	
		
	\end{lemma}
	
	Wang \textit{et al.}\cite{WangZhou2025} gave a series of LPNFs and constructed a series of LAZ sequence sets, which are presented in the above theorem. Some of them are asymptotically optimal in both periodic and aperiodic cases. In the next section, we will reveal some new constructions of LPNFs,
		which provides significant improvements over existing frameworks \cite{WangZhou2025}. And then produce a series of LAZ sequence sets with new parameters that are not covered by the parameters of \cite{WangZhou2025}. Furthermore, some classes of the proposed sequence sets demonstrates an asymptotically optimal property in both periodic and aperiodic cases.
	
	\section{Proposed Constructions of LPNFs}
	
	For $K\geq3$, let $\mathbb{Z}_K^\ast$ denote the multiplicative group of
	all invertible elements of the ring $\mathbb{Z}_K$. Then $\mathbb{Z}_K^\ast$ contains $\varphi(K)$  elements, where $\varphi(\cdot)$ is the Euler's total function. In this section, our primary objective is to construct three new classes of LPNFs.  The following lemma provides a necessary foundation for the subsequent constructions of LPNFs.
\begin{lemma}
   Let $K=p_1^{e_1}p_2^{e_2}\cdots p_s^{e_s}$, where $s\geq1$, the $p_i$ are distinct odd primes and $e_i\geq1$ with $1\leq i\leq s$. Denote $E=\textrm{lcm}\{\varphi(p_i^{e_i})=p_i^{e_i-1}(p_i-1):1\leq i\leq s\}=NR$ with $N\geq2$, $N,R\in\mathbb{Z}$. Then there is an element $\alpha$ of order $N$ in $\mathbb{Z}_K^\ast$. 
\end{lemma}
\begin{proof}
    By the Chinese remainder theorem, we have an isomorphism of rings:
		\begin{equation*}\label{PAF}
			\mathbb{Z}_K\cong\mathbb{Z}_{p_1^{e_1}}\times\mathbb{Z}_{p_2^{e_2}}\times\cdots\times\mathbb{Z}_{p_s^{e_s}},\,\,\,\,\,\,x\mapsto (x_1,\cdots x_s),
		\end{equation*}
		where $x_i\equiv x\,(\bmod\,p_i^{e_i})$, $1\leq i\leq s$, which also keeps the isomorphism of groups
		\begin{equation*}\label{PAF}
			\mathbb{Z}_K^\ast\cong\mathbb{Z}_{p_1^{e_1}}^\ast\times\mathbb{Z}_{p_2^{e_2}}^\ast\times\cdots\times\mathbb{Z}_{p_s^{e_s}}^\ast,
		\end{equation*}
		we identify $x$ with the image $(x\,(\bmod\,p_1^{e_1}),x\,(\bmod\,p_2^{e_2})\cdots x\,(\bmod\,p_i^{e_i}))$.
		
		For every $i$, $1\leq i\leq s$, the group $\mathbb{Z}_{p_i^{e_i}}^\ast$ is cyclic of order $\varphi(p_i^{e_i})$, so that $\mathbb{Z}_{p_i^{e_i}}^\ast=\langle\pi_i\rangle$ with $\pi_i$ as a generator. On the other hand, $\pi=(\pi_1,\pi_2,\cdots,\pi_s)$ is an element of $\mathbb{Z}_K^\ast$ with order $E=\textrm{lcm}\{\varphi(p_i^{e_i}):1\leq i\leq s\}$. Hence, there exist a cyclic subgroup $\langle\pi\rangle=\{1,\pi,\cdots,\pi^{E-1}\}$ of $\mathbb{Z}_K^\ast$ with order $E$. Since $E=NR$, $\alpha=\pi^R$ is an element of $\mathbb{Z}_K^\ast$ of order $N$. This completes the proof.
\end{proof}
%
	\begin{theorem}\label{NewLPNF}
    Let $K=p_1^{e_1}p_2^{e_2}\cdots p_s^{e_s}$, where $s\geq1$, the $p_i$ are distinct odd primes and $e_i\geq1$ with $1\leq i\leq s$. Denote $E=\textrm{lcm}\{\varphi(p_i^{e_i})=p_i^{e_i-1}(p_i-1):1\leq i\leq s\}=NR$ with $N\geq2$, $N,R\in\mathbb{Z}$. Define the function \( f: \mathbb{Z}_N \rightarrow \mathbb{Z}_K \) by $f(x) = \alpha^x$, where \( \alpha \in \mathbb{Z}_K^\ast \) is an element of order \( N \). Then we have
		
		(i)\,\,the function $f(x)=\alpha^x$ is 
        $\langle N,K,Z_x,Z_y\rangle$-LPNF, where $Z_x=\min\{\frac{p_i-1}{\gcd(p_i-1,R)}:1\leq i\leq s\}$ and $Z_y=K$.
		
		(ii)\,\,the function $f(x)=\alpha^x$ is $\langle N,K,Z_x,Z_y\rangle$-LPNF, where $Z_x=N$ and $Z_y=p$, $p$ is the smallest prime factor of $K$.
	\end{theorem}
	
	\begin{proof}
		 Recall the function 
		$$f:\,\,\,\mathbb{Z}_N \rightarrow \mathbb{Z}_K,\,\,\,\,\,f(x)=\alpha^x,$$
        where $\alpha$ is an element of order $N$. 
        
		(i)\,\,For $a\in\mathbb{Z}_N\setminus\{0\}$ and $b\in\mathbb{Z}_K$, consider a solution $x$ in $\mathbb{Z}_N$
		of equation
		\begin{equation*}
			f(x+a)-f(x)=b,
		\end{equation*}
		then we have
		\begin{equation}\label{equaf}
			\alpha^{x+a}-\alpha^{x}=\alpha^{x}(\alpha^{a}-1)=b,
		\end{equation}
		Note that $\alpha^{a}-1\in\mathbb{Z}_K^\ast$ means $\alpha^{a}-1$ is invertible in $\mathbb{Z}_K$. Furthermore, we can obtain the following equivalence relation :
		\begin{eqnarray}\label{equ}
			\alpha^{a}-1\in\mathbb{Z}_K^\ast&\Longleftrightarrow&\alpha^{a}-1\in\mathbb{Z}_{p_i^{e_i}}^\ast=\langle\pi_i\rangle \,\,\,\,\,\,(\mathrm{for\,all}\,i,\,1\leq i\leq s)\nonumber \\
			&\Longleftrightarrow&\alpha^{a}-1\not\equiv 0\,(\bmod\,p_i)\,\,\,\,\,\,(\mathrm{since}\, x\in\mathbb{Z}_{p_i^{e_i}}^\ast\Longleftrightarrow p_i\nmid x)\,\,\,\,\,\,(1\leq i\leq s)\nonumber \\
			&\Longleftrightarrow&\pi^{aR}\not\equiv 1\,(\bmod\,p_i)\,\,\,\,\,\,(1\leq i\leq s)\nonumber \\&\Longleftrightarrow& p_i-1\nmid aR\,\,\,\,\,\,(\mathrm{since\,the\,order\,of}\,\pi\,(\bmod\,p_i)\,\mathrm{is}\,(p_i-1))\,\,\,\,\,\,(1\leq i\leq s)\nonumber \\
			&\Longleftrightarrow&\frac{p_i-1}{\gcd(p_i-1,R)}\nmid a\,\,\,\,\,\,(1\leq i\leq s).
		\end{eqnarray}
		
		Therefore, from the equivalence relation we conclude that if $-Z_x<a<Z_x$ and $a\neq0$, where $Z_x={\text{min}}\{\frac{p_i-1}{\gcd(p_i-1,R)}:1\leq i\leq s\}$, then $\alpha^{a}-1\in\mathbb{Z}_K^\ast$ and (\ref{equaf})  can be transformed into  $\alpha^{x}=b(\alpha^{a}-1)^{-1}$. Hence, if $b(\alpha^{a}-1)^{-1}\in\langle\alpha\rangle$, where  $\langle\alpha\rangle=\{1,\alpha,\cdots,\alpha^{N-1}\}$ is a cyclic  subgroup of $\mathbb{Z}_K^\ast$ of order $N$, then (\ref{equaf}) must have a unique solution $x\in\mathbb{Z}_N$. If $b(\alpha^{a}-1)^{-1}\notin\langle\alpha\rangle$, then (\ref{equaf}) has no solution in $\mathbb{Z}_N$.
        
        (ii)\,\,For any integers $-N<a<N$ with $a\neq0$ and $-p<b<p$, consider the number of solutions to
		\begin{equation}\label{equa}
			f(x+a)-f(x)\equiv b\,(\bmod\,K)
		\end{equation}
		over $x\in\mathbb{Z}_N$, where $p$ denotes the smallest prime factor of $K$. 
		Next, we distinguish the two cases $b=0$ and $b\neq0$.
		
		If $b=0$, then we get $\alpha^{x}(\alpha^{a}-1)\equiv 0\,(\bmod\,K)$. Since $\alpha\in\mathbb{Z}_K^\ast$, we have $\alpha^x\in\mathbb{Z}_K^\ast$, and thus $\alpha^{a}-1\equiv 0\,(\bmod\,K)$, i.e., $\alpha^{a}\equiv 1\,(\bmod\,K)$. Recall that the order of $\alpha$ mod $K$ is $N$, so $N\,|\,a$, which contradicts to $-N<a<N$ with $a\neq0$. Therefore, the equation (\ref{equa}) has no solution in $\mathbb{Z}_N$.     
		
		If $b\neq0$, then we have $\alpha^{x}(\alpha^{a}-1)\equiv b\,(\bmod\,K)$. Because $p$ is the smallest prime factor of $K$, for $1\leq i\leq s$, all $p_i$ hold $p_i\geq p$ and $-p_i<b<p_i$. 
        By combining $b\neq0$, we get $p_i\nmid b$, that is, $b\not\equiv 0\,(\bmod\,p_i)$, which yields $b\in \mathbb{Z}_{p_i^{e_i}}^\ast$ from $ \mathbb{Z}_{p_i^{e_i}}^\ast= \mathbb{Z}_{p_i^{e_i}}\setminus p_i\mathbb{Z}_{p_i^{e_i}}$. 
        Applying the Chinese remainder theorem, we obtain that $b$ is invertible in $\mathbb{Z}_K$, i.e., $b\in \mathbb{Z}_{K}^\ast$. Note that $\alpha^x\in\mathbb{Z}_K^\ast$, which leads to $\alpha^{a}-1\equiv \alpha^{-x}b\,(\bmod\,K)$, and hence, it also implies $\alpha^a-1\in\mathbb{Z}_K^\ast$. Together with $\alpha^{x}=b(\alpha^{a}-1)^{-1}\in \mathbb{Z}_{K}^\ast$ and the order of $\alpha$ mod $K$ is $N$, we conclude that equation (\ref{equa}) admits a unique solution $x\in\mathbb{Z}_N$.
		
        Consequently, the equation (\ref{equa}) has at most one solution.
       This completes the proof of Theorem \ref{NewLPNF}.
	\end{proof}
    
     \begin{example}
		Let $K=5^2=25$, $N=20$, and we choose an element $\alpha=2$ of order $20$ to define the function
		$$f:\,\,\,\mathbb{Z}_{20} \rightarrow \mathbb{Z}_{25},\,\,\,\,\,f(x)=2^x.$$
		Then one can get
        \begin{table}[htp]
			\centering
			\begin{tabular}{ccccccccccccccccccccc}
				\Xhline{1pt}
				x  & 0& 1 & 2 & 3 & 4 & 5& 6&	7&	8&	9&	10&	11&	12&	13&	14&	15&	16&	17&	18&	19\\
				\Xhline{1.5pt}
				$f(x)$ &1&	2	&4&	8	&16	&7&	14&	3&	6&	12&	24&	23&	21&	17&	9&	18	&11&	22&	19	&13\\
				\Xhline{1pt}
			\end{tabular}
		\end{table}
        
		For $Z_x={\text{min}}\{\frac{p_i-1}{\gcd(p_i-1,R)}:1\leq i\leq s\}=4$, $Z_y=K=20$ and take $a\in(-\mathbb{Z}_{4},\mathbb{Z}_{4})\setminus\{0\}$, we have
		\begin{equation}\nonumber\label{lpnf}
			\left\{\begin{array}{ll}
            \{f(x-3)-f(x):0\leq x\leq19\}=\{21, 17, 9, 18, 11, 22, 19, 13, 1, 2, 4, 8, 16, 7, 14, 3, 6, 12, 24, 23\} \\
				\{f(x-2)-f(x):0\leq x\leq19\}=\{18, 11, 22, 19, 13, 1, 2, 4, 8, 16, 7, 14, 3, 6, 12, 24, 23, 21, 17, 9\} \\
				 \{f(x-1)-f(x):0\leq x\leq19\}=\{12, 24, 23, 21, 17, 9, 18, 11, 22, 19, 13, 1, 2, 4, 8, 16, 7, 14, 3, 6\}\\
				\{f(x+1)-f(x):0\leq x\leq19\}=\{1, 2, 4, 8, 16, 7, 14, 3, 6, 12, 24, 23, 21, 17, 9, 18, 11, 22, 19, 13\} \\
				\{f(x+2)-f(x):0\leq x\leq19\}=\{3, 6, 12, 24, 23, 21, 17, 9, 18, 11, 22, 19, 13, 1, 2, 4, 8, 16, 7, 14\} \\
				\{f(x+3)-f(x):0\leq x\leq19\}=\{7, 14, 3, 6, 12, 24, 23, 21, 17, 9, 18, 11, 22, 19, 13, 1, 2, 4, 8, 16\} 
					\end{array}\right.
		\end{equation}
		These indicate that $f(x)=2^x$ is a $\langle 20,25,4,20\rangle$-LPNF, which aligns with our Theorem \ref{NewLPNF} (i). 
	\end{example}
    \begin{example}
		Let $K=3^2=9$, $N=6$, and we choose an element $\alpha=2$ of order $6$ to define the function
		$$f:\,\,\,\mathbb{Z}_{6} \rightarrow \mathbb{Z}_{9},\,\,\,\,\,f(x)=2^x.$$
		Then one can get
        \begin{table}[htp]
			\centering
			\begin{tabular}{ccccccc}
				\Xhline{1pt}
				x  & 0& 1 & 2 & 3 & 4 & 5\\
				\Xhline{1.5pt}
				$f(x)$ &1&	2	&4&	8	&7	&5\\
				\Xhline{1pt}
			\end{tabular}
		\end{table}
		
        For $Z_x=N=6$, $Z_y=p=3$ and take $a\in(-\mathbb{Z}_{6},\mathbb{Z}_{6})\setminus\{0\}$, we have
		\begin{equation}\nonumber\label{lpnf}
			\left\{\begin{array}{ll}
				\{f(x-5)-f(x):0\leq x\leq5\}=\{f(x+1)-f(x):0\leq x\leq5\}=\{1, 2, 4, 8, 7, 5\} \\
				\{f(x-4)-f(x):0\leq x\leq5\}=\{f(x+2)-f(x):0\leq x\leq5\}=\{3, 6, 3, 6, 3, 6\} \\
				\{f(x-3)-f(x):0\leq x\leq5\}=\{f(x+3)-f(x):0\leq x\leq5\}=\{7, 5, 1, 2, 4, 8\}\\
                \{f(x-2)-f(x):0\leq x\leq5\}=\{f(x+4)-f(x):0\leq x\leq5\}=\{6, 3, 6, 3, 6, 3\} \\
               \{f(x-1)-f(x):0\leq x\leq5\}=\{f(x+5)-f(x):0\leq x\leq5\}=\{4, 8, 7, 5, 1, 2\} 
					\end{array}\right.
		\end{equation}
		We find that the solutions of $f(x+4)-f(x)=3$ are $1,3,5$. Moreover, $f(x-4)-f(x)=3$, $f(x-2)-f(x)=3$, and $f(x+2)-f(x)=3$ have more than one solution. These imply that $f(x)=2^x$ is a $\langle 6,9,6,3\rangle$-LPNF, which aligns with our Theorem \ref{NewLPNF} (ii). 

	\end{example}
	
	The next characterizes a construction of LPNFs with new parameters, which are employed to construct LAZ sequence sets.
   \begin{construction}\label{cons2}
        Let $K=p_1^{e_1}p_2^{e_2}\cdots p_s^{e_s}$, where $s\geq1$, the $p_i$ are distinct odd primes and $e_i\geq1$ with $1\leq i\leq s$. Denote $E=\textrm{lcm}\{\varphi(p_i^{e_i})=p_i^{e_i-1}(p_i-1):1\leq i\leq s\}=NR$ with $N\geq2$, $N,R\in\mathbb{Z}$. Put $K_1\geq K$, we will define the function $f:\,\mathbb{Z}_N \rightarrow \mathbb{Z}_{K_1}$ as 
        $f(x)=\langle\alpha^x\rangle_K$, 
        where $\alpha$ is an element of order $N$, $\langle\alpha\rangle_K$ is the least positive integer of $\alpha$ modulo $K$.
    \end{construction}
	\begin{theorem}\label{NewLPNF2}
		Let $f$ be a function defined in Construction \ref{cons2}.  
        Then we have $f$ 
is an $\langle N,K_1,Z_x,Z_y\rangle$-LPNF, where $Z_x={\text{min}}\{\frac{p_i-1}{\gcd(p_i-1,R)}:1\leq i\leq s\}$ and $Z_y=K_1-K+2$.
        
	\end{theorem}
	\begin{proof}
		By the definition of LPNFs, it suffices to consider the equation 
        \begin{equation}\label{equa2}
        f(x+a)-f(x)=\langle\alpha^{x+a}\rangle_K-\langle\alpha^{x}\rangle_K\equiv b\,(\bmod\,K),
        \end{equation}
        which has at most one solution $x\in\mathbb{Z}_N$ when $-Z_x<a<Z_x$ with $a\neq0$ and $-Z_y<b<Z_y$.\\
        Since $1\leq f(x),f(x+a)\leq K-1$, we must have $$-(K-2)\leq f(x+a)-f(x)\leq K-2.$$ Thus, from $-(K_1-K+2)< b< K_1-K+2$ yields $$-(K_1-1)\leq f(x+a)-f(x)-b\leq K_1-1.$$
        Furthermore, the equation (\ref{equa2}) implies $K_1\,|\,(f(x+a)-f(x)-b)$, and so, it denotes the conventional arithmetic operations applied to integers $\mathbb{Z}$, i.e.,
		\begin{equation}\label{Z_k}
			f(x+a)-f(x)-b=\langle\alpha^{x+a}\rangle_K-\langle\alpha^{x}\rangle_K-b=0.
		\end{equation}
        Consequently, the equation (\ref{Z_k}) is valid in $\mathbb{Z}_K$, that is,
        $$f(x+a)-f(x)=\langle\alpha^{x+a}\rangle_K-\langle\alpha^{x}\rangle_K\equiv b\,(\bmod\,K).$$
        But $f(x)=\langle\alpha^{x}\rangle_K\equiv \alpha^{x}\,(\rm mod\,K)$,
         so that we can get 
         \begin{equation}\label{equa22}
           \alpha^{x+a}-\alpha^{x}\equiv b\,(\bmod\,K).  
         \end{equation}
        For $-{\min}\{\frac{p_i-1}{\gcd(p_i-1,R)}:1\leq i\leq s\}<a<{\min}\{\frac{p_i-1}{\gcd(p_i-1,R)}:1\leq i\leq s\}$ with $a\neq 0$, applying Theorem \ref{NewLPNF}, we can immediately conclude that $\alpha^{a}-1\in\mathbb{Z}_K^\ast$. Hence, equation (\ref{equa22}) has at most a solution in $\mathbb{Z}_N$ if $-{\min}\{\frac{p_i-1}{\gcd(p_i-1,R)}:1\leq i\leq s\}<a<{\min}\{\frac{p_i-1}{\gcd(p_i-1,R)}:1\leq i\leq s\}$ and $-(K_1-K+2)<b<(K_1-K+2)$.

		Consequently, the equation (\ref{equa2}) admits at most one solution.
	\end{proof}
   
     \begin{example}
		Let $K=7$, $K_1=10$, $N=6$, and we choose an element $\alpha=5$ of order $6$ to define the function
		$$f:\,\,\,\mathbb{Z}_{6} \rightarrow \mathbb{Z}_{10},\,\,\,\,\,f(x)=\langle5^x\rangle_7.$$
		Then one can get
        \begin{table}[h]
			\centering
			\begin{tabular}{ccccccc}
				\Xhline{1pt}
				x  & 0& 1 & 2 & 3 & 4 & 5\\
				\Xhline{1.5pt}
				$f(x)$ &1&	5	&4&	6	&2	&3\\
				\Xhline{1pt}
			\end{tabular}
		\end{table}
		
        For $Z_x={\text{min}}\{\frac{p_i-1}{\gcd(p_i-1,R)}:1\leq i\leq s\}=6$, $Z_y=K_1-K+2=5$ and take $a\in(-\mathbb{Z}_{6},\mathbb{Z}_{6})\setminus\{0\}$, we have
		\begin{equation}\nonumber\label{lpnf}
			\left\{\begin{array}{ll}
				\{f(x-5)-f(x):0\leq x\leq5\}=\{f(x+1)-f(x):0\leq x\leq5\}=\{4, 9, 2, 6, 1, 8\} \\
				 \{f(x-4)-f(x):0\leq x\leq5\}=\{f(x+2)-f(x):0\leq x\leq5\}=\{3, 1, 8, 7, 9, 2\} \\
				\{f(x-3)-f(x):0\leq x\leq5\}=\{f(x+3)-f(x):0\leq x\leq5\}=\{5, 7, 9, 5, 3, 1\}\\
                \{f(x-2)-f(x):0\leq x\leq5\}=\{f(x+4)-f(x):0\leq x\leq5\}=\{1, 8, 7, 9, 2, 3\} \\
                \{f(x-1)-f(x):0\leq x\leq5\}=\{f(x+5)-f(x):0\leq x\leq5\}=\{2, 6, 1, 8, 4, 9\} 
					\end{array}\right.
		\end{equation}
		These show that the solutions of $f(x+3)-f(x)=5$ and $f(x-3)-f(x)=5$ are $0,3$. we conclude that $f(x)=2^x$ is a $\langle 6,10,6,5\rangle$-LPNF, which aligns with our Theorem \ref{NewLPNF2}. 

	\end{example}
		\section{Proposed LAZ sequence sets}
	By applying the LPNFs in Theorems \ref{NewLPNF} and \ref{NewLPNF2} to the above Lemma \ref{2025sset}, we directly get the following LAZ sequence sets. Meanwhile, we also analyze optimality of these LAZ sequence sets.
	
	\begin{theorem}\label{Newsset}
		Let $K=p_1^{e_1}p_2^{e_2}\cdots p_s^{e_s}$, where $s\geq1$, $\{p_1,p_2,\cdots,p_s\}$ be distinct odd primes and $e_i\geq1$ with $1\leq i\leq s$, $E=\textrm{lcm}\{\varphi(p_i^{e_i}):1\leq i\leq s\}=NR$ with $N\geq2$. Then: 
        
        (i)\,\,A periodic LAZ sequence set $\mathcal{S}$ with parameters $(M,D,\Pi,K)$ from Theorem \ref{NewLPNF} (i), where $M=N$ is the number of sequences, $D=NK$ is the length of the sequences, $\Pi=(-Z_x,Z_x)\times(-Z_y,Z_y)$, $Z_x=\min\{\frac{p_i-1}{\gcd(p_i-1,R)}:1\leq i\leq s\}$, and $Z_y=K$.

        (ii)\,\,A periodic LAZ sequence set $\mathcal{S}$ with parameters $(N,NK,\Pi,K)$ from Theorem \ref{NewLPNF} (ii), where $\Pi=(-Z_x,Z_x)\times(-Z_y,Z_y)$, $Z_x=N$, $Z_y=p$, and $p$ is the smallest prime factor of $K$.
	\end{theorem}
    \begin{proof}
The results follow directly by applying the LPNFs given in Theorem~\ref{NewLPNF} to Lemma~\ref{2025sset}, with the corresponding parameters $(M, D, \Pi, K)$ and $(N,NK,\Pi,K)$ specified in Theorem~\ref{Newsset} (i) and (ii).
\end{proof}



    Next, we present a specific LAZ sequence set in Theorem \ref{Newsset}.
     \begin{corollary}\label{cor-p-ssetp66}
		Let $K=p^{e}$, where $p$ is an odd prime, $e\geq1$, and $N=\varphi(p^{e})=p^{e-1}(p-1)$. Let $\mathcal{S}$ denote a periodic LAZ sequence set from Theorem \ref{Newsset} above. 
        
        (i)\,\,The $\mathcal{S}$ is an $(N,NK,\Pi,K)$ with $\Pi=(-Z_x,Z_x)\times(-Z_y,Z_y)$, $Z_y=K=p^{e}$, $Z_x=p-1$. 

        (ii)\,\,The $\mathcal{S}$ is an $(N,NK,\Pi,K)$ with $\Pi=(-Z_x,Z_x)\times(-Z_y,Z_y)$, $Z_y=p$, $Z_x=N=p^{e-1}(p-1)$. 
		
	\end{corollary}
    Next, we discuss the optimality of the periodic LAZ sequence set presented in Corollary \ref{cor-p-ssetp66}.
 \begin{theorem}\label{optimalc1}
     Let $\mathcal{S}$ denote a periodic LAZ sequence set from Corollary \ref{cor-p-ssetp66} above. Then
     
     (i)\,\,the $\mathcal{S}$ presented in Corollary \ref{cor-p-ssetp66} (i) is asymptotically optimal with regard to the YZFLLT bound in Lemma \ref{2013bound} (i);
     
     (ii)\,\,the $\mathcal{S}$ presented in Corollary \ref{cor-p-ssetp66} (ii) is asymptotically optimal with regard to the YZFLLT bound in Lemma \ref{2013bound} (i).
    \end{theorem}
    \begin{proof}
		The parameters $M$, $D$, $Z_x$, $Z_y$, $\theta_{\text{max}}$ and $\hat{\theta}_{\text{max}}$ can be seen directly from Theorem \ref{2025sset} and \ref{Newsset} by taking $s=R=1$. By Theorem \ref{2013bound}, the YZFLLT bound $\Delta$ can be
		derived as follows:
        \begin{eqnarray}
			\Delta&=&\frac{D}{\sqrt{Z_y}}\sqrt{\frac{MZ_xZ_y-D}{D(MZ_x-1)}}\nonumber \\&=&\sqrt{\frac{D^2MZ_xZ_y-D^3}{DZ_y(MZ_x-1)}} \nonumber \\ 
            &=&\sqrt{\frac{D(MZ_xZ_y-D)}{Z_yMZ_x-Z_y}} \nonumber \\
			&=&\sqrt{\frac{D(1-\frac{D}{MZ_xZ_y})}{1-\frac{1}{MZ_xZ_y}}}.\label{bound1c}
		\end{eqnarray}
        
        (i)\,\,Recall $M=N$, $D=NK$, $Z_y=K=p^{e}$, $Z_x=p-1$, then (\ref{bound1c}) yields
        \begin{eqnarray}
			\Delta
            &=&\sqrt{\frac{D(1-\frac{1}{p-1})}{1-\frac{1}{p^ep^{e-1}(p-1)^2}}}\label{bound1ic}.
		\end{eqnarray}
          
        Consider that if $p\rightarrow\infty$, then $\frac{K}{N}=\frac{p^e}{p^{e-1}(p-1)}=\frac{1}{1-\frac{1}{p}}\rightarrow1$. Together with (\ref{bound1ic}) and $\theta_{\text{max}}=K$, hence we infer 
        $$\lim\limits_{p \to \infty}\frac{\theta_{\text{max}}}{\Delta}=\lim\limits_{p \to \infty}\sqrt{\frac{K^2}{D}}=\lim\limits_{p \to \infty}\sqrt{\frac{K}{N}}=1.$$   
        
        (ii)\,\,Recall $M=N$, $D=NK$, $Z_y=p$, $Z_x=N=p^{e-1}(p-1)$, it is easily checked that $\Delta$ is actually equal to (\ref{bound1ic}),
         and so 
        $$\lim\limits_{p \to \infty}\frac{\theta_{\text{max}}}{\Delta}=\lim\limits_{p \to \infty}\sqrt{\frac{K^2}{D}}=\lim\limits_{p \to \infty}\sqrt{\frac{K}{N}}=1.$$
        The proof of Theorem \ref{optimalc1} is completed.
	\end{proof}
    	Similarly, we will present some other common scenarios of $K$ in the Remark, i.e. $K$ is the product of twin primes or distinct odd primes.
    \begin{remark}
		Case A: When $K=p(p+2)$, where $p$ and $p+2$ are odd prime, and suppose $N=\textrm{lcm}(p-1,p+1)=\frac{p^2-1}{2}$. 
        
        (i)\,\,There exist a periodic LAZ sequence sets with parameters $(N,NK,\Pi,K)$, where $\Pi=(-Z_x,Z_x)\times(-Z_y,Z_y)$, $Z_y=K=p^{e}$, $Z_x=p-1$. Moreover, $\lim\limits_{D \to \infty} \rho_{\text{LAZ}}=\sqrt{2}$, where $\rho_{\text{LAZ}}=\frac{\theta_{\text{max}}}{\Delta}$.
        
        (ii)\,\,There exist a periodic LAZ sequence sets with parameters $(N,NK,\Pi,K)$, where $\Pi=(-Z_x,Z_x)\times(-Z_y,Z_y)$, $Z_y=p$, $Z_x=N=p^{e-1}(p-1)$. Moreover, $\lim\limits_{D \to \infty} \rho_{\text{LAZ}}=\sqrt{2}$, where $\rho_{\text{LAZ}}=\frac{\theta_{\text{max}}}{\Delta}$.

        Case B: When $K=p_1p_2\cdots p_s$, where $p_1<p_2<\cdots<p_s$ are odd primes, and suppose $N=E=\textrm{lcm}\{\varphi(p_i)=p_i-1:1\leq i\leq s\}=\frac{(p_1-1)(p_2-1)\cdots(p_s-1)}{2}$. 
        
        (i)\,\,There exist a periodic LAZ sequence set with parameters $(N,NK,\Pi,K)$, where $\Pi=(-Z_x,Z_x)\times(-Z_y,Z_y)$, $Z_y=K=p_1p_2\cdots p_s$, $Z_x=p_1-1$. Moreover, $\lim\limits_{D \to \infty} \rho_{\text{LAZ}}=\sqrt{2}$ and $\lim\limits_{D \to \infty} \hat{\rho}_{\text{LAZ}}=\sqrt{2}$, where $\rho_{\text{LAZ}}=\frac{\theta_{\text{max}}}{\Delta}$ and $\hat{\rho}_{\text{LAZ}}=\frac{\hat{\theta}_{\text{max}}}{\hat{\Delta}}$.
        
        (ii)\,\,There exist a periodic LAZ sequence set with parameters $(N,NK,\Pi,K)$, where $\Pi=(-Z_x,Z_x)\times(-Z_y,Z_y)$, $Z_y=p_1$, $Z_x=N$. Moreover, $\lim\limits_{D \to \infty} \rho_{\text{LAZ}}=\sqrt{2}$, where $\rho_{\text{LAZ}}=\frac{\theta_{\text{max}}}{\Delta}$.
	\end{remark}
    The following Theorem presents the aperiodic LAZ set obtained via Theorem \ref{NewLPNF} and Lemma~\ref{2025sset}.
    \begin{theorem}\label{Newssetap}
		Let $K=p_1^{e_1}p_2^{e_2}\cdots p_s^{e_s}$, where $s\geq1$, $\{p_1,p_2,\cdots,p_s\}$ be distinct odd primes and $e_i\geq1$ with $1\leq i\leq s$, $E=\textrm{lcm}\{\varphi(p_i^{e_i}):1\leq i\leq s\}=NR$ with $N\geq2$. We obtain: 
        
        (i)\,\,An aperiodic LAZ sequence set $\mathcal{S}$ with parameters $(N,NK,\Pi,K+Z_x-1)$ from Theorem \ref{NewLPNF} (i), where $\Pi=(-Z_x,Z_x)\times(-Z_y,Z_y)$, $Z_x=\min\{\frac{p_i-1}{\gcd(p_i-1,R)}:1\leq i\leq s\}$, $Z_y=K$.

        (ii)\,\,An aperiodic LAZ sequence set $\mathcal{S}$ with parameters $(N,NK,\Pi,K+Z_x-1)$ from Theorem \ref{NewLPNF} (ii), where $\Pi=(-Z_x,Z_x)\times(-Z_y,Z_y)$, $Z_x=N$, $Z_y=p$, $p$ is the smallest prime factor of $K$.
	\end{theorem}
\begin{proof}
The results follow directly by applying the LPNFs given in Theorem~\ref{NewLPNF} to Lemma~\ref{2025sset}, with the corresponding parameters $(M, D, \Pi, K+Z_x-1)$ and $(N,NK,\Pi,K+Z_x-1)$ specified in Theorem~\ref{Newssetap} (i) and (ii).
\end{proof}
   
	    \begin{corollary}\label{cor-p-ssetap66}
		Let $K=p^{e}$, where $p$ is an odd prime, $e\geq1$, and $N=\varphi(p^{e})=p^{e-1}(p-1)$. Let $\mathcal{S}$ denote an aperiod LAZ sequence set from Theorem \ref{Newssetap} above. Then
        
        (i)\,\, the $\mathcal{S}$ is an $(N,NK,\Pi,K+p-2)$ with $\Pi=(-Z_x,Z_x)\times(-Z_y,Z_y)$, $Z_y=K=p^{e}$, $Z_x=p-1$. 

        (ii)\,\,the $\mathcal{S}$ is an $(N,NK,\Pi,K+N-1)$ with $\Pi=(-Z_x,Z_x)\times(-Z_y,Z_y)$, $Z_y=p$, $Z_x=N=p^{e-1}(p-1)$. 
		
	\end{corollary}
    
    The next Theorem analyzes the optimality of the aperiodic case.   
     \begin{theorem}\label{optimalc2}
     Let $\mathcal{S}$ denote an aperiodic LAZ sequence set from Corollary \ref{cor-p-ssetap66} above. Then
     
     (i)\,\,the $\mathcal{S}$ presented in Corollary \ref{cor-p-ssetap66} (i) is asymptotically optimal under the YZFLLT bound in Lemma \ref{2013bound} (ii) if $e\geq 2$;
     
     (ii)\,\,the optimality factor of the $\mathcal{S}$ presented in Corollary \ref{cor-p-ssetap66} (ii) is $\lim\limits_{D \to \infty} \hat{\rho}_{\text{LAZ}}=2$ under the YZFLLT bound in Lemma \ref{2013bound} (ii), where  $\hat{\rho}_{\text{LAZ}}=\frac{\hat{\theta}_{\text{max}}}{\hat{\Delta}}$.
    \end{theorem}
    \begin{proof}
		The parameters $M$, $D$, $Z_x$, $Z_y$, $\theta_{\text{max}}$ and $\hat{\theta}_{\text{max}}$ can be seen directly from Theorem \ref{2025sset} and \ref{Newsset} by taking $s=R=1$. By Theorem \ref{2013bound}, the YZFLLT bound $\Delta$ can be
		derived as follows:
        \begin{eqnarray}
			\hat{\Delta}&=&\frac{D}{\sqrt{Z_y}}\sqrt{\frac{MZ_xZ_y-D-Z_x+1}{(MZ_x-1)(D+Z_x-1)}}\nonumber \\&=&\sqrt{\frac{D^2(MZ_xZ_y-D-Z_x+1)}{Z_y(MZ_x-1)(D+Z_x-1)}}\nonumber \\
            &=&\sqrt{\frac{D(MZ_xZ_y-D-Z_x+1)}{Z_y(MZ_x-1)(1+\frac{Z_x}{D}-\frac{1}{D})}}\nonumber \\
			&=&\sqrt{\frac{D(1-\frac{D}{MZ_xZ_y}-\frac{1}{MZ_y}+\frac{1}{MZ_xZ_y})}{(1-\frac{1}{MZ_x})(1+\frac{Z_x}{D}-\frac{1}{D})}}.\label{bound2c}
		\end{eqnarray}
        
        (i)\,\,Recall $M=N$, $D=NK$, $Z_y=K=p^{e}$, $Z_x=p-1$, then (\ref{bound2c}) yields
        \begin{eqnarray}
			\hat{\Delta}
			&=&\sqrt{\frac{D(1-\frac{1}{p-1}-\frac{1}{p^ep^{e-1}(p-1)}+\frac{1}{p^ep^{e-1}(p-1)^2})}{(1-\frac{1}{p^{e-1}(p-1)^2})(1+\frac{1}{p^ep^{e-1}}-\frac{1}{p^ep^{e-1}(p-1)})}}.\label{bound2ic}
		\end{eqnarray}
        we observe that if $e\geq2$ and $p\rightarrow\infty$, then (\ref{bound2ic}) implies 
        $$\lim\limits_{p \to \infty}\frac{\hat{\theta}_{\text{max}}}{\hat{\Delta}}=\lim\limits_{p \to \infty}\sqrt{\frac{(K+p-2)^2}{D}}=\lim\limits_{p \to \infty}\sqrt{\frac{(K+p-2)^2}{NK}}=1,$$ since $\hat{\theta}_{\text{max}}=K+Z_x-1=p^e+p-2\rightarrow p^e=K$ and $\frac{K}{N}=\frac{p^e}{p^{e-1}(p-1)}\rightarrow1$.  
        
        (ii)\,\,Recall $M=N$, $D=NK$, $Z_y=p$, $Z_x=N=p^{e-1}(p-1)$,  (\ref{bound2c}) produces
         \begin{eqnarray}
			\hat{\Delta}
			&=&\sqrt{\frac{D(1-\frac{1}{p-1}-\frac{1}{p^e(p-1)}+\frac{1}{p^ep^{e-1}(p-1)^2})}{(1-\frac{1}{p^{2e-2}(p-1)^2})(1+\frac{1}{p^e}-\frac{1}{p^ep^{e-1}(p-1)})}}.\label{bound2iic}
		\end{eqnarray}
        But since $\hat{\theta}_{\text{max}}=K+Z_x-1=K+N-1$, 
        combining (\ref{bound2iic}) and $\frac{K}{N}=\frac{p^e}{p^{e-1}(p-1)}\rightarrow1$, we obtain
        \begin{eqnarray}
			\lim\limits_{p \to \infty}\frac{\hat{\theta}_{\text{max}}}{\hat{\Delta}}&=&\sqrt{\frac{K^2+(N-1)^2+2K(N-1)}{NK}}\nonumber \\&=&\sqrt{\frac{K^2}{NK}+\frac{N^2}{NK}-\frac{2N}{NK}+\frac{1}{NK}+\frac{2KN}{NK}-\frac{2K}{NK}}\nonumber \\
            &=&\sqrt{\frac{K}{N}+\frac{N}{K}+2}\nonumber \\
            &=&2
		\end{eqnarray}
        The proof of Theorem \ref{optimalc2} is completed.
	\end{proof}
    \begin{remark}
		Case A: When $K=p(p+2)$, where $p$ and $p+2$ are odd prime, and suppose $N=\textrm{lcm}(p-1,p+1)=\frac{p^2-1}{2}$. There exist an aperiodic LAZ sequence sets with parameters $(N,NK,\Pi,K)$, where $\Pi=(-Z_x,Z_x)\times(-Z_y,Z_y)$, $Z_y=K=p^{e}$, $Z_x=p-1$. Moreover, $\lim\limits_{D \to \infty} \hat{\rho}_{\text{LAZ}}=\sqrt{2}$, where $\hat{\rho}_{\text{LAZ}}=\frac{\hat{\theta}_{\text{max}}}{\hat{\Delta}}$.

        Case B: When $K=p_1p_2\cdots p_s$, where $p_1<p_2<\cdots<p_s$ are odd primes, and suppose $N=E=\textrm{lcm}\{\varphi(p_i)=p_i-1:1\leq i\leq s\}=\frac{(p_1-1)(p_2-1)\cdots(p_s-1)}{2}$. There exist an aperiodic LAZ sequence set with parameters $(N,NK,\Pi,K+p_1-2)$, where $\Pi=(-Z_x,Z_x)\times(-Z_y,Z_y)$, $Z_y=K=p_1p_2\cdots p_s$, $Z_x=p_1-1$, $\theta_{\text{max}}=K=p_1p_2\cdots p_s$. Moreover, $\lim\limits_{D \to \infty} \hat{\rho}_{\text{LAZ}}=\sqrt{2}$, where $\hat{\rho}_{\text{LAZ}}=\frac{\hat{\theta}_{\text{max}}}{\hat{\Delta}}$.
	\end{remark}

    In the above, we discuss the periodic and aperiodic LAZ sequence set by new LPNFs in Theorems \ref{NewLPNF} and their optimality. The following Theorem, we will state a series of periodic and aperiodic LAZ sequence set from Theorem \ref{NewLPNF2}.
	\begin{theorem}\label{Newsset3}
		Let $K=p_1^{e_1}p_2^{e_2}\cdots p_s^{e_s}$, where $s\geq1$, $\{p_1,p_2,\cdots,p_s\}$ be distinct odd primes and $e_i\geq1$ with $1\leq i\leq s$,  $E=\textrm{lcm}\{\varphi(p_i^{e_i}):1\leq i\leq s\}=NR$ with $N\geq2$. Then, from Theorem \ref{NewLPNF2}, we have:
        
        (i)\,\,A periodic LAZ sequence set $\mathcal{S}$ with parameters $(N,NK,\Pi,K_1)$, where $\Pi=(-Z_x,Z_x)\times(-Z_y,Z_y)$, $Z_y=K_1-K+2$, $K\leq K_1$, $Z_x=\min\{\frac{p_i-1}{\gcd(p_i-1,R)}:1\leq i\leq s\}$.

        (ii)\,\,An aperiodic LAZ sequence set $\mathcal{S}$ with parameters $(N,NK,\Pi,K_1+Z_x-1)$, where $\Pi=(-Z_x,Z_x)\times(-Z_y,Z_y)$, $Z_y=K_1-K+2$, $K\leq K_1$, $Z_x=\min\{\frac{p_i-1}{\gcd(p_i-1,R)}:1\leq i\leq s\}$.
	\end{theorem}
\begin{proof}
The results follow directly by applying the LPNFs given in Theorem~\ref{NewLPNF2} to Lemma~\ref{2025sset}, with the corresponding parameters $(N, NK, \Pi, K_1)$  and $(N, NK, \Pi, K_1+Z_x-1)$ specified in Theorem~\ref{Newsset3} (i) and (ii).
\end{proof}

    From the Throrem \ref{Newsset3}, the following result is immediate.
    \begin{corollary}\label{cor-p-sset3c}
		Let $K=p^{e}$, where $p$ is an odd prime, $e\geq1$, and let $N=\varphi(p^{e})=p^{e-1}(p-1)$. Then we obtain:
        
        (i)\,\,A periodic LAZ sequence set $\mathcal{S}$ with parameters $(N,NK_1,\Pi,K_1)$, where $\Pi=(-Z_x,Z_x)\times(-Z_y,Z_y)$, $Z_y=K_1-K+2$, $Z_x=p-1$.
        
        (ii)\,\,An aperiodic LAZ sequence set $\mathcal{S}$ with parameters $(N,NK_1,\Pi,K_1+p-2)$, where $\Pi=(-Z_x,Z_x)\times(-Z_y,Z_y)$, $Z_y=K_1-K+2$, $Z_x=p-1$.
	\end{corollary}
    \begin{theorem}\label{optimal3c}
    Let $\mathcal{S}$ denote a LAZ sequence set from Corollary \ref{cor-p-sset3c} above. And let $\lim\limits_{p \to \infty}\sqrt{\frac{K_1}{N}}=1$ and $\lim\limits_{p \to \infty} \frac{K_1}{(p-1)(K_1-K+2)}=0$. Then
     
     (i)\,\,the $\mathcal{S}$ attains asymptotically optimal with respect to the periodic YZFLLT bound in Lemma \ref{2013bound} (i);
     
     (ii)\,\,the $\mathcal{S}$ attains asymptotically optimal with respect to the aperiodic YZFLLT bound in Lemma \ref{2013bound} (ii) if $e\geq2$.
    \end{theorem}
    \begin{proof}
        By (\ref{bound1c}) and (\ref{bound2c}), we obtain the YZLFLT bounds $\Delta$ and $\hat{\Delta}$: 
\begin{equation}\label{bounds0}
			\left\{\begin{array}{ll}
				\Delta=\frac{D}{\sqrt{Z_y}}\sqrt{\frac{MZ_xZ_y-D}{D(MZ_x-1)}}=\sqrt{\frac{D(1-\frac{D}{MZ_xZ_y})}{1-\frac{1}{MZ_xZ_y}}}, \\
				\hat{\Delta}=\frac{D}{\sqrt{Z_y}}\sqrt{\frac{MZ_xZ_y-D-Z_x+1}{(MZ_x-1)(D+Z_x-1)}}=\sqrt{\frac{D(1-\frac{D}{MZ_xZ_y}-\frac{1}{MZ_y}+\frac{1}{MZ_xZ_y})}{(1-\frac{1}{MZ_x})(1+\frac{Z_x}{D}-\frac{1}{D})}}.
			\end{array}\right.
		\end{equation}
 
 (i)\,\,We use $\Delta$ in (\ref{bounds0}) with $M=N=p^{e-1}(p-1)$, $D=NK_1$, $Z_y=K_1-K+2$, $Z_x=p-1$ to get
 \begin{eqnarray}
			\Delta
            &=&\sqrt{\frac{D(1-\frac{K_1}{(p-1)(K_1-K+2)})}{1-\frac{1}{p^{e-1}(p-1)^2(K_1-K+2)}}}\label{bound1i}.
		\end{eqnarray}
        Moreover $\theta_{\text{max}}=K_1$, $\lim\limits_{p \to \infty}\sqrt{\frac{K_1}{N}}=1$, and $\lim\limits_{p \to \infty} \frac{K_1}{(p-1)(K_1-K+2)}=0$, it follows that $$\lim\limits_{p \to \infty}\frac{\theta_{\text{max}}}{\Delta}=\lim\limits_{p \to \infty}\sqrt{\frac{K_1^2}{D}}=\lim\limits_{p \to \infty}\sqrt{\frac{K_1}{N}}=1.$$ 

		(ii)\,\,We use $\hat{\Delta}$ in (\ref{bounds0}) with $M=N=p^{e-1}(p-1)$, $D=NK_1$, $Z_y=K_1-K+2$, $Z_x=p-1$ to get
        \begin{eqnarray}
			\hat{\Delta}
            &=&\sqrt{\frac{D(1-\frac{K_1}{(p-1)(K_1-K+2)}-\frac{1}{p^{e-1}(p-1)(K_1-K+2)}+\frac{1}{p^{e-1}(p-1)^2(K_1-K+2)})}{(1-\frac{1}{p^{e-1}(p-1)^2})(1+\frac{1}{p^{e-1}K_1}-\frac{1}{p^{e-1}(p-1)K_1})}}\label{bound1iend}.
		\end{eqnarray}
We observe that if $e\geq2$ and $p\rightarrow\infty$, combining $\hat{\theta}_{\text{max}}=K_1+p-2$, we have 

$$\lim\limits_{p \to \infty}\frac{\hat{\theta}_{\text{max}}}{\hat{\Delta}}=\lim\limits_{p \to \infty}\sqrt{\frac{(K_1)^2+2K_1(p-2)+(p-2)^2}{NK_1}}=\lim\limits_{p \to \infty}\sqrt{\frac{K_1}{N}}=1$$ 
since $\lim\limits_{p \to \infty}\sqrt{\frac{K_1}{N}}=1$ and $\lim\limits_{p \to \infty} \frac{K_1}{(p-1)(K_1-K+2)}=0$.
The proof of Theorem \ref{optimal3c} is complete.
    \end{proof}

	For the $(N,NK,\Pi,\theta_{\text{max}}\,\,or\,\,\hat{\theta}_{\text{max}})$-LAZ sequence set constructed in \cite{WangZhou2025}, the value of the parameter $N$ is only odd. But for the construction in Theorem \ref{NewLPNF} (i), the value of parameter $N$, as a factor of $\textrm{lcm}\{\varphi(p_i^{e_i}):1\leq i\leq s\}$, can be chosen to be even. Thus the LAZ sequence set in Theorem \ref{Newsset} has more general and flexible parameters. Moreover, we also present two other new classes of constructions of LPNFs in Theorem \ref{NewLPNF} (ii) and Theorem \ref{NewLPNF2}, which are used to design more LAZ sequence sets with flexible parameters. 

	\subsection{Example of Periodic and Aperiodic LAZ sequence set}
    In this subsection, we present two examples to illustrate our proposed LAZ sequence sets in Theorems \ref{Newsset}, \ref{Newssetap} and \ref{Newsset3}.

	\begin{example}\label{ex2}
	    Let $K=25$, $N=20$, $\alpha=2$ be the element in $\mathbb{Z}_K$ of order $20$. we define the function
		$$f:\,\,\,\mathbb{Z}_{20} \rightarrow \mathbb{Z}_{25},\,\,\,\,\,f(x)=2^x.$$
		Then	
		\begin{table}[h]
			\centering
			\begin{tabular}{ccccccccccccccccccccc}
				\Xhline{1pt}
				x  & 0& 1 & 2 & 3 & 4 & 5& 6&	7&	8&	9&	10&	11&	12&	13&	14&	15&	16&	17&	18&	19\\
				\Xhline{1.5pt}
				$f(x)$ &1&	2	&4&	8	&16	&7&	14&	3&	6&	12&	24&	23&	21&	17&	9&	18	&11&	22&	19	&13\\
				\Xhline{1pt}
			\end{tabular}
		\end{table}
        
The sequence set \( \mathcal{S} \) is created using Theorem \ref{Newsset}. It can be demonstrated that \( \mathcal{S} \) forms a set of periodic sequences \( (20, 500, \Pi_1, 25) \)-LAZ and \((20, 500, \Pi_2, 25)\)-LAZ, where \( \Pi_1 = (-4, 4) \times (-25, 25) \) and \( \Pi_2 = (-20, 20) \times (-5, 5) \). A glimpse of the periodic auto-AF and cross-AF of the sequences in \( \mathcal{S} \) can be seen in Fig. \ref{fig:examplep2}. 
        \begin{figure}[htp]
			\centering
			\includegraphics[width=1\linewidth]{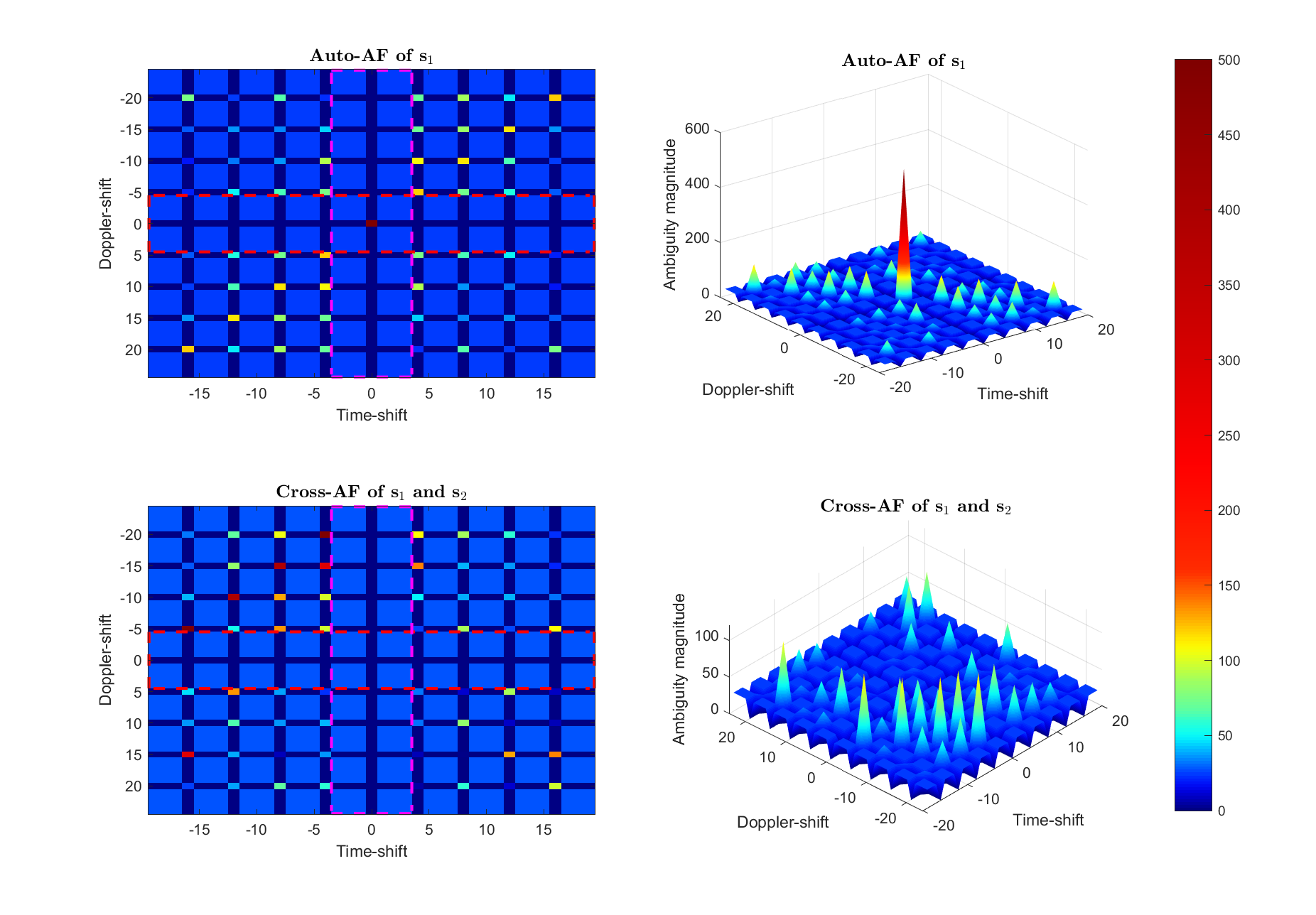}
			\caption{A glimpse of the periodic auto-AF and cross-AF of the sequence set \( \mathcal{S} \) in Example \ref{ex2}.}
			\label{fig:examplep2}
		\end{figure}
        
        Furthermore, it can be shown that \( \mathcal{S} \) also constitutes an aperiodic set of \(  (20, 500, \Pi_1, 30) \)-LAZ and \((20, 500, \Pi_2, 30)\)-LAZ sequences, where \( \Pi_1 = (-4, 4) \times (-25, 25) \) and \( \Pi_2 = (-20, 20) \times (-5, 5) \). A glimpse of the aperiodic auto-AF and cross-AF of the sequences in \( \mathcal{S} \) can be seen in Fig. \ref{fig:exampleap2}. 
		\begin{figure}[htp]
			\centering
			\includegraphics[width=0.9\linewidth]{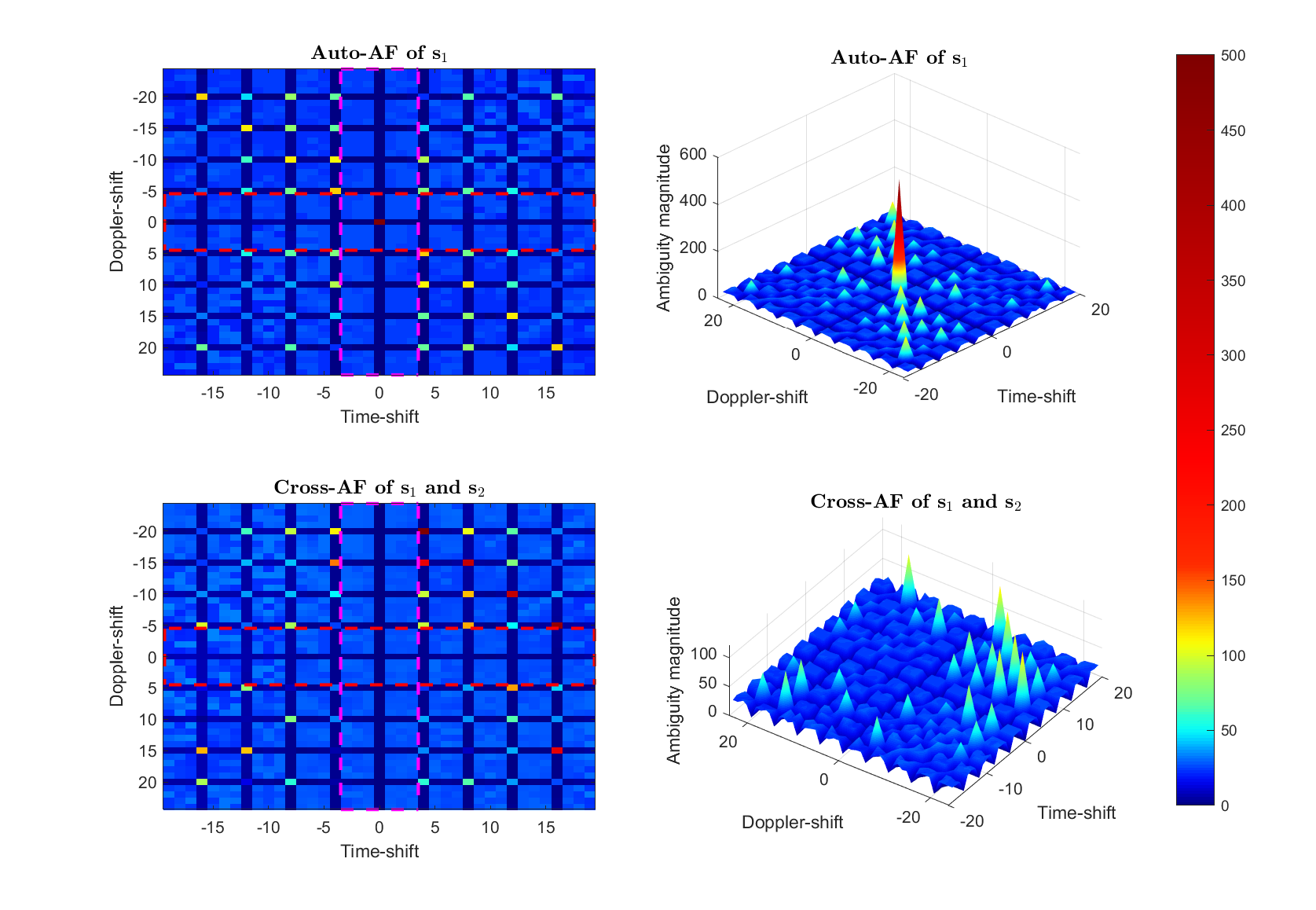}
			\caption{A glimpse of the aperiodic auto-AF and cross-AF of the sequence set $\mathcal{S}$ in Example \ref{ex2}.}
			\label{fig:exampleap2}
		\end{figure}
	\end{example}

	\begin{example}\label{ex3}
		Let $K=7$, $K_1=10$, $N=6$, $\alpha=5$ be the element in $\mathbb{Z}_K$ of order $6$.  we define the function
		$$f:\,\,\,\mathbb{Z}_{6} \rightarrow \mathbb{Z}_{10},\,\,\,\,\,f(x)=\langle 5^x \rangle_{K}.$$
        Regard $f(x)$ as a function in $\mathbb{Z}_{K_1}$. Then	
		\begin{table}[htp]
			\centering
			\begin{tabular}{ccccccc}
				\Xhline{1pt}
				x  & 0& 1 & 2 & 3 & 4 & 5 \\
				\Xhline{1.5pt}
				$f(x)$ & 1 &5 & 4 & 6 & 2 & 3\\
				\Xhline{1pt}
			\end{tabular}
		\end{table}
		
		The sequence set \( \mathcal{S} \) is created using Theorem \ref{Newsset3}. It can be demonstrated that \( \mathcal{S} \) forms a periodic \( (6, 60, \Pi, 10) \)-LAZ sequence set, where \( \Pi = (-6, 6) \times (-5, 5) \). A glimpse of the periodic auto-AF and cross-AF of the sequences in \( \mathcal{S} \) can be seen in Fig. \ref{fig:examplep3}. 
		\begin{figure}[htp]
			\centering
			\includegraphics[width=0.9\linewidth]{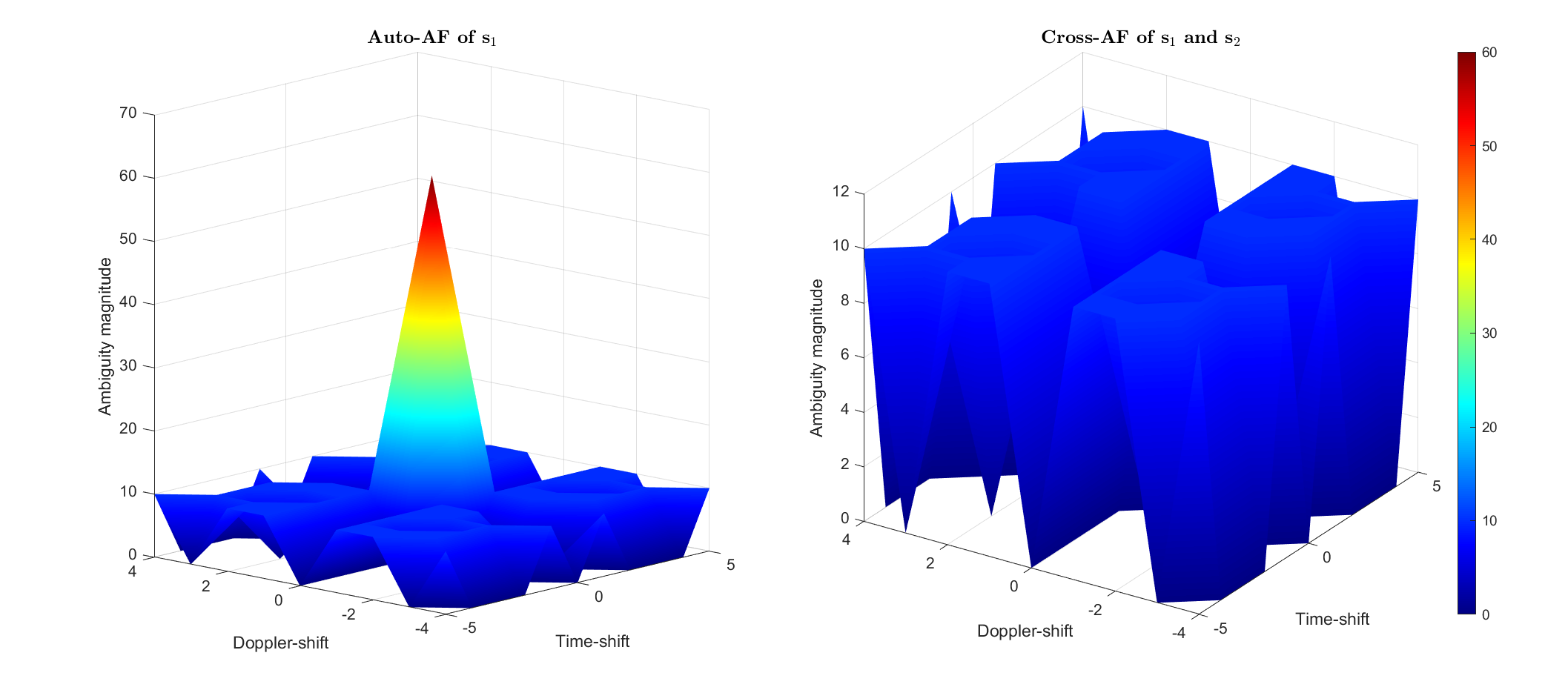}
			\caption{A glimpse of the periodic auto-AF and cross-AF of the sequence set \( \mathcal{S} \) in Example \ref{ex3}.}
			\label{fig:examplep3}
		\end{figure}
        
		Additionally, it can be shown that \( \mathcal{S} \) also constitutes an aperiodic \( (6, 60, \Pi, 15) \)-LAZ sequence set, where \( \Pi = (-6, 6) \times (-5, 5) \). A glimpse of the aperiodic auto-AF and cross-AF of the sequences in \( \mathcal{S} \) can be seen in Fig. \ref{fig:exampleap3}. 
		\begin{figure}[htp]
			\centering
			\includegraphics[width=0.9\linewidth]{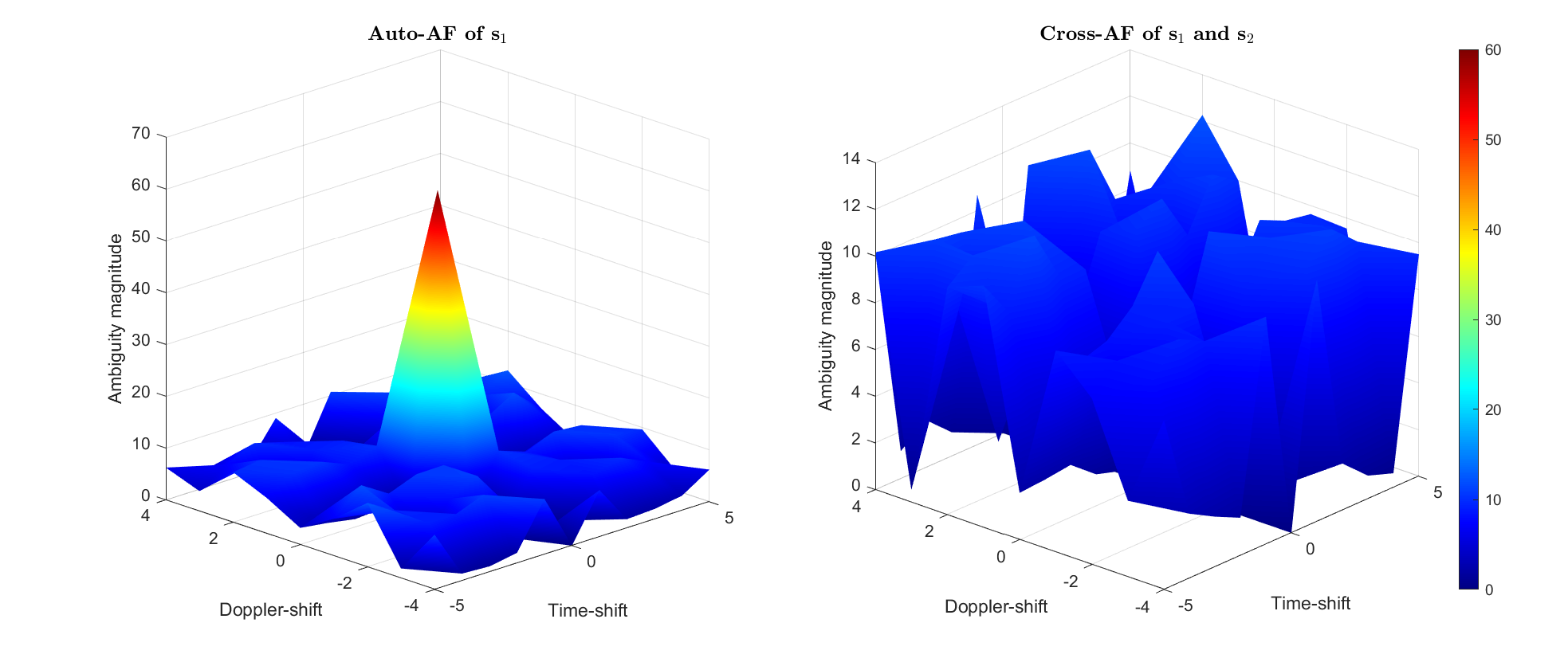}
			\caption{A glimpse of the aperiodic auto-AF and cross-AF of the sequence set $\mathcal{S}$ in Example \ref{ex3}.}
			\label{fig:exampleap3}
		\end{figure}
        

	\end{example}	
	\subsection{Parameters of the Derived Periodic and Aperiodic LAZ Sequence Sets}
	In Tables \ref{th1}, \ref{table12}, and \ref{table2}, we display several parameters of the proposed periodic LAZ sequence sets in Theorems \ref{Newsset} and \ref{Newsset3}, along with their corresponding optimality factor, $\rho_{\text{LAZ}}.$ Furthermore, we provide parameters of the proposed aperiodic LAZ sequence sets in Theorems \ref{Newssetap} and \ref{Newsset3}, along with their corresponding optimality factor, \(\hat{\rho}_{\text{LAZ}}\), in Tables \ref{2} and \ref{3}, respectively. It is obvious from these tables that the optimality factors ${\rho}_{\mathrm{LAZ}}$ and \(\hat{\rho}_{\text{LAZ}}\) of the constructed LAZ sequence sets are asymptotically approaching $1$. Thus, with careful selection of parameters, our sequences are strong candidates for ISAC system.
	
	\begin{table}[htp]
		\caption{Parameters of the proposed periodic LAZ sequence set in Theorem \ref{Newsset}(i) } 
		\label{th1}
		\begin{center}
			\begin{tabular}{|c|c|c|c|c|}
				\hline Set size  & Length  & Low ambiguity zone &  Maximum periodic AF & Optimality factor  \\
				$N$ & $NK$ &  $\Pi$ &  magnitude $\theta_{\max}$ &  $\rho_{\mathrm{LAZ}}$ \\
				\hline 6 & 54 & $(-3,3) \times(-9,9)$ & 9 & 1.4577 \\
				\hline  20 & 500 & $(-5,5) \times(-25,25)$ & 25 & 1.2437 \\
				\hline  42 & 2058 & $(-7,7) \times(-49,49)$ & 49 & 1.1647\\
				\hline  110 & 13310 & $(-11,11) \times(-121,121)$ & 121 & 1.0995 \\
				\hline  156 & 26364 & $(-13,13) \times(-169,169)$ & 169 &  1.0831 \\
				\hline  272 & 78608 & $(-17,17) \times(- 289, 289)$ & 289 &  1.0624  \\
				\hline  506 & 267674 & $(-23,23) \times(-529,529)$ & 529 &1.0454 \\
				\hline  812 & 682892 & $(-29,29) \times(-841,841)$ & 841 & 1.0357 \\
				\hline  930 & 893730 & $(-31,31) \times(-961,961)$ & 961 & 1.0333 \\
				\hline  1332 &1823508 & $(-37,37) \times(-1369,1369)$ & 1369 &  1.0278 \\
				\hline
			\end{tabular}
		\end{center}
	\end{table}
	
	\begin{table}[htp]
		\caption{Parameters of the proposed periodic LAZ sequence set in Theorem \ref{Newsset}(ii)}
		\label{table12}
		\begin{center}
			\begin{tabular}{|c|c|c|c|c|}
				\hline Set size  & Length  & Low ambiguity zone &  Maximum periodic AF & Optimality factor  \\
				$N$ & $NK$ &  $\Pi$ &  magnitude $\theta_{\max}$ &  $\rho_{\mathrm{LAZ}}$ \\
				\hline 6 & 54 & $(-6,6) \times(-3,3)$ & 9 & 1.7078 \\
				\hline  20 & 500 & $(-20,20) \times(-5,5)$ & 25 & 1.2894 \\
				\hline  42 & 2058 & $(-42,42) \times(-7,7)$ & 49 & 1.1829 \\
				\hline  110 & 13310 & $(-110,110) \times(-11,11)$ & 121 & 1.1055 \\
				\hline  156 & 26364 & $(-156,156) \times(-13,13)$ & 169 & 1.0871 \\
				\hline  272 & 78608 & $(-272,272) \times(-17,17)$ & 289 & 1.0646  \\
				\hline  506 & 267674 & $(-506,506) \times(-23,23)$ & 529 &1.0465 \\
				\hline  812 & 682892 & $(-812,812) \times(-29,29)$ & 841 & 1.0364 \\
				\hline  930 & 893730 & $(-930,930) \times(-31,31)$ & 961 &  1.0339 \\
				\hline  1332 &1823508 & $(-1332,1332) \times(-37,37)$ & 1369 & 1.0282 \\
				\hline
			\end{tabular}
		\end{center}
	\end{table}
	
	\begin{table}[htp]
		\caption{Parameters of the proposed periodic LAZ sequence set in Theorem \ref{Newsset3}}
		\label{table2}
		\begin{center}
			\begin{tabular}{|c|c|c|c|c|}
				\hline Set size  & Length  & Low ambiguity zone &  Maximum periodic AF & Optimality factor  \\
				$N$ & $NK_1$ &  $\Pi$ &  magnitude $\theta_{\max}$ &  $\rho_{\mathrm{LAZ}}$ \\
				\hline  6 & 54 & $(-6,6) \times(-4,4)$ & 9 & 1.5275 \\
				\hline  12 & 204 & $(-12,12) \times(-6,6)$ & 17 & 1.3571  \\
				\hline  22 & 638 & $(-22,22) \times(-8,8)$ & 29 &1.2550 \\
				\hline 36 & 1584 &	$(-36,36) \times(-9,9)$& 44 & 1.1888  \\
				\hline  52 & 3224 & $(-52,52) \times(-11,11)$ & 62 &1.1562 \\
				\hline 66 & 5082 &	$(-66,66) \times(-12,12)$& 77 & 1.1367 \\
				\hline  78 & 7020 & $(-78,78) \times(-13,13)$ & 90 &1.1252 \\
				\hline  96 & 10464 & $(-96,96) \times(-14,14)$ & 109 & 1.1115 \\
				\hline 108 & 13176 &	$(-108,108) \times(-15,15)$&122 & 1.1052  \\
				\hline 126 & 17766 &	$(-126,126) \times(-16,16)$& 141 & 1.0969  \\
				\hline
			\end{tabular}
		\end{center}
	\end{table}	
 
	
	\begin{table}[htp]
		\caption{Parameters of the proposed aperiodic LAZ sequence set Theorem \ref{Newsset}(i)}
		\label{2}
		\begin{center}
			\begin{tabular}{|c|c|c|c|c|}
				\hline Set size  & Length  & Low ambiguity zone &  Maximum periodic AF & Optimality factor  \\
				$N$ & $NK$ &  $\Pi$ &  magnitude $\theta_{\max}$ &  $\rho_{\mathrm{LAZ}}$ \\
				\hline 6 & 54 & $(-3,3) \times(-9,9)$ & 11 & 1.8314 \\
				\hline  20 & 500 & $(-5,5) \times(-25,25)$ & 29 & 1.4499 \\
				\hline  42 & 2058 & $(-7,7) \times(-49,49)$ & 55 & 1.3095\\
				\hline  110 & 13310 & $(-11,11) \times(-121,121)$ & 131 &1.1909 \\
				\hline  156 & 26364 & $(-13,13) \times(-169,169)$ & 181 &  1.1603\\
				\hline  272 & 78608 & $(-17,17) \times(- 289, 289)$ & 305 &  1.1213  \\
				\hline  506 & 267674 & $(-23,23) \times(-529,529)$ & 551 &1.0889 \\
				\hline  812 & 682892 & $(-29,29) \times(-841,841)$ & 869 & 1.0702 \\
				\hline  930 & 893730 & $(-31,31) \times(-961,961)$ & 991 & 1.0656 \\
				\hline  1332 &1823508 & $(-37,37) \times(-1369,1369)$ & 1405 &  1.0548 \\
				\hline
			\end{tabular}
		\end{center}
	\end{table}	
	
	\begin{table}[htp]
		\caption{Parameters of the proposed aperiodic LAZ sequence set Theorem \ref{Newsset3}}
		\label{3}
		\begin{center}
			\begin{tabular}{|c|c|c|c|c|}
				\hline Set size  & Length  & Low ambiguity zone &  Maximum aperiodic AF & Optimality factor  \\
				$N$ & $NK_1$ &  $\Pi$ &  magnitude $\hat{\theta}_{\max}$ &  $\hat{\rho}_{\mathrm{LAZ}}$ \\
				\hline  42 & 3990 & $(-6,6) \times(-48,48)$ & 100 & 1.9319 \\
				\hline  156 & 62244 & $(-12,12) \times(-232,232)$ & 410 &1.7752 \\
				\hline  506 & 445786 & $(-22,22) \times(-354,354)$ & 902 & 1.4345 \\
				\hline  1332 & 2521476 & $(-36,36) \times(-526,526)$ & 1928 & 1.2798  \\
				\hline  2756 & 9725924 & $(-52,52) \times(-722,722)$ & 3580 & 1.2060 \\
				\hline  4422 & 23790360 & $(-66,66) \times(-893,893)$ & 5445 & 1.1711 \\
				\hline  6162 & 44853198 & $(-78,78) \times(-1040,1040)$ & 7356 &1.1512 \\
				\hline 9312 & 99340416 & $(-96,96) \times(-1261,1261)$ & 10763 & 1.1308 \\
				\hline  11772 & 156402792 & $(-108,108) \times(-1407,1407)$ & 13393 & 1.1210\\
				\hline 16002 & 284275530 & $(-126,126) \times(-1638,1638)$ & 17890 & 1.1099  \\
				\hline
			\end{tabular}
		\end{center}
	\end{table}

	\section{Concluding Remarks}\label{conclusion} 
	We investigated three constructions of LPNFs with new parameters. By leveraging their structural features, we further proposed a series of periodic and aperiodic LAZ sequence sets with more flexible parameters. In comparison with prior works, we provided more new LAZ sequence sets, owing to reduced constraints on parameters. Significantly, two classes of these newly developed periodic and aperiodic LAZ sequence sets are asymptotically optimal based on YZFLLT bounds. Moreover, another class of periodic LAZ sequence sets is asymptotically optimal based on YZFLLT bounds.  

	\bibliographystyle{ieeetr}
	\bibliography{thebibliography1.bib}

\end{document}